\documentclass[a4paper,UKenglish]{lipics}

\usepackage{microtype}

\usepackage{amsmath}
\usepackage{amscd}
\usepackage{amssymb}
\usepackage{subfig}
\usepackage{verbatim}
\usepackage{graphicx}
\usepackage{url}
\usepackage{color}
\usepackage{tikz}
\usepackage{subfig}

\newtheorem{lem}{Lemma}
\newtheorem{thm}{Theorem}
\def\Real{\Bbb R}
\def\pnorm#1#2{\|\,#1 \,\|_{#2}}
\def\calP{{\mathcal{P}}}
\def\calR{{\mathcal{R}}}
\def\calF{{\mathcal{F}}}
\def\dist#1#2{\mbox{dist}(#1,#2)}
\def\ts#1{\mathfrak{S}_{#1}}
\def\sensitivity#1#2{\sigma_{#1}(#2)}
\def\abs#1{{\lvert #1 \rvert}}
\def\sensitivity#1#2{\sigma_{#1}(#2)}
\def\dim#1{{\rm dim}\left(#1\right)}
\def\proj#1#2{{\rm proj}\left(#1,#2\right)}
\def\innerproduct#1#2{\langle #1,#2 \rangle}
\def\spn#1{{\rm span}\left(#1\right)}
\def\vdist#1#2{{\rm dist}_{\vert}\left(#1,#2\right)}

\title{On the Sensitivity of Shape Fitting Problems\footnote{This material is based upon
 work supported by the National Science Foundation under Grant No. 0915543.}}
\author[1]{Kasturi Varadarajan}
\author[2]{Xin Xiao}
\affil[1]{Department of Computer Science\\
University of Iowa, Iowa City, IA 52242, USA\\
\texttt{kasturi-varadarajan@uiowa.edu}}
\affil[2]{Department of Computer Science\\
University of Iowa, Iowa City, IA 52242, USA\\
\texttt{xin-xiao@uiowa.edu}}
\authorrunning{Kasturi Varadarajan and Xin Xiao} 

\Copyright[nc-nd]
          {Kasturi Varadarajan and Xin Xiao}


\serieslogo{}
\volumeinfo
  {Billy Editor, Bill Editors}
  {2}
  {Conference title on which this volume is based on}
  {1}
  {1}
  {1}
\EventShortName{}
\DOI{10.4230/LIPIcs.xxx.yyy.p}

\begin{document}

\maketitle

\begin{abstract}
In this article, we study shape fitting problems, $\epsilon$-coresets, and total sensitivity. We focus on the $(j,k)$-projective clustering problems, including $k$-median/$k$-means, $k$-line clustering, $j$-subspace approximation, and the integer $(j,k)$-projective clustering problem. We derive upper bounds of total sensitivities for these problems, and obtain $\epsilon$-coresets using these upper bounds. Using a dimension-reduction type argument, we are able to greatly simplify earlier results on total sensitivity for the $k$-median/$k$-means clustering problems, and obtain positively-weighted $\epsilon$-coresets for several variants of the $(j,k)$-projective clustering problem. We also extend an earlier result on $\epsilon$-coresets for the integer $(j,k)$-projective clustering problem in fixed dimension to the case of high dimension.
\end{abstract}

\section{Introduction}
In this article, we study shape fitting problem, coresets, and in particular, total sensitivity. A shape fitting problem is specified by a triple $(\Real^d,\calF, {\rm dist} )$, where $\Real^d$ is the $d$-dimensional Euclidean space, $\calF$ is a family of subsets of $\Real^d$, and ${\rm dist}: \Real^d \times \Real^d \to \Real^+$ is a continuous function that we will refer to as a {\em distance} function. We also assume that (a) $\dist{p}{q} = 0$ if and only if $p = q$, and (b)
$\dist{p}{q} = \dist{q}{p}$. We refer to each $F \in {\cal F}$ as a {\em shape}, and we require each shape $F$ to be a non-empty, closed, subset of $\Real^d$. We define the {\em distance} of a point $p \in \Real^d$ to a shape $F \in {\cal F}$ to be 
$\dist{p}{F} = \min_{q \in F}\dist{p}{q}$. An instance of a shape fitting problem is specified by a finite point set $P \subset \Real^d$. We slightly abuse notation and use $\dist{P}{F}$ to denote $\sum_{p \in P}\dist{p}{F}$ when $P$ is a set of points in $\Real^d$. The goal is to find a shape which best fits $P$, that is, a shape minimizing $\sum_{p \in P} \dist{p}{F}$ over all shapes $F \in \calF$. This is referred to as the $L_1$ fitting problem, which is the main focus of this paper. In the $L_{\infty}$ fitting problem, we seek to find a shape $F \in \calF$ minimizing $\max_{p \in P}\dist{p}{F}$.

In this paper, we focus on the $(j,k)$-projective clustering problem. Given non-negative integers $j$ and $k$, the family of shapes is the set of $k$-tuples of affine $j$-subspaces (that is, $j$-flats) in $\Real^d$. More precisely, each shape is the union of some $k$ $j$-flats. The underlying distance function is usually the $z^{\rm th}$ power of the Euclidean distance, for a positive real number $z$. When $j=0$, $\calF$ is the set of all $k$-point sets of $\Real^d$, so the $(0,k)$-projective clustering problem is the $k$-median clustering problem when the distance function is the Euclidean distance, and it is the $k$-means clustering problem when the distance function is the square of the Euclidean distance; when $j=1$, the family of shapes is the set of $k$-tuples of lines in $\Real^d$; when $k=1$, $(j,1)$-projective clustering is the subspace approximation problem, where the family of shapes is the set of $j$-flats. Other than these projective clustering problems where $j$ or $k$ is set to specific values,
 another variant of the $(j,k)$-projective clustering problem is the integer $(j,k)$-projective clustering problem, where we assume that the input points have integer coordinates (but there is no restriction on $j$ and $k$), and the magnitude of these coordinates is at most $n^c$, where $n$ is the number of input points and $c>0$ is some constant. That is, the points are in a polynomially large integer grid.

An $\epsilon$-coreset for an instance $P$ of a shape fitting problem is a weighted set $S$, such that for any shape $F \in \calF$, the summation of distances from points in $P$ approximates the weighted summation of the distances from points in $S$ up to a  multiplicative factor of $(1\pm \epsilon)$. A more precise definition (Definition~\ref{defn: epsiloncoreset}) follows later. Coresets can be considered as a succinct representation of the point set; in particular, in order to obtain a $(1+\epsilon)$-approximation solution fitting $P$, it is sufficient to find a $(1+\epsilon)$-approximation solution for the coreset $S$. One usually seeks a small coreset, whose {\em size} $|S|$ is independent of the cardinality of $P$. Coresets of size $o(n)$ for the $(j,k)$-projective clustering problem for general $j$ and $k$ are not known to exist. However, the $k$-median/$k$-means clustering, $k$-line clustering, $j$-subspace approximation, and integer $(j,k)$-projective clustering problems admit small coresets.

Langberg and Schulman~\cite{DBLP:conf/soda/LangbergS10} introduced a general approach to coresets via  the notion of {\em sensitivity} of points in a point set, which provides a natural way to set up a probability distribution $\Pr{\cdot}$ on $P$. Roughly speaking, the sensitivity of a point with respect to a point set measures the importance of the point, in terms of fitting shapes in the given family of shapes $\calF$. Formally, the sensitivity of point $p$ in a point set $P$ is defined by $\sensitivity{P}{p}:=\sup_{F \in \calF} \dist{p}{F}/\dist{P}{F}$. (In the degenerate case where the denominator in the ratio is $0$, the numerator is also $0$, and we take the ratio to be $0$; the reader should feel free to ignore this technicality.) The total sensitivity of a point set $P$ is defined by $\ts{P}:=\sum_{p \in P}\sensitivity{P}{p}$. The nice property of quantifying the ``importance'' of a point in a point set is that for any $F \in \calF$, $\dist{p}{F}/\dist{P}{F} \leq \sensitivity{P}{p}$. Setting the 
probability of selecting $p$ to be $\sensitivity{P}{p}/\ts{P}$, and the weight of $p$ to be $\ts{P}/\sensitivity{P}{p}$, $\forall p \in P$, one can show that the variance of the sampling 
scheme is $O((\ts{P})^2)$. When $\ts{P}$ is $o(n)$, (for example, a constant or logarithmic in terms of $n = |P|$), one can obtain an $\epsilon$-coreset by sampling a small number of points. Langberg and Schulman~\cite{DBLP:conf/soda/LangbergS10} show that the total sensitivity of any (arbitrarily large) point set $P \subset \Real^d$ for $k$-median/$k$-means clustering problem is a constant, depending only on $k$, independent of the cardinality of $P$ and the dimension of the Euclidean space where $P$ and $\calF$ are from. Using this, they derived a coreset for these problems with size depending polynomially on $d$ and $k$ and independent of $n$. Their work can be seen as evolving from earlier work on coresets for the $k$-median/$k$-means and related problems via other low variance sampling schemes \cite{1655320, Clarkson05subgradientand, 1247072, DBLP:journals/siamcomp/DasguptaDHKM09}. 

Feldman and Langberg~\cite{DBLP:conf/stoc/FeldmanL11} relate the notion of an $\epsilon$-coreset with the well-studied notion of an $\epsilon$-approximation of range spaces. They use a ``functional representation'' of points: consider a family of functions $\calP=\{f_p(\cdot) \vert p \in P\}$, where each point $p$ is associated with a function $f_p: X \to \Real$. The target here is to pick a small subset $S \subseteq P$ of points, and assign weights appropriately, so that $\sum_{p \in S}w_p f_p(x)$ approximates $\sum_{p \in P}f_p(x)$ at every $x \in X$. When $X$ is $\calF$ and $f_p(F)=\dist{p}{F}$, this is just the original $\epsilon$-coreset for $P$. However, $f_p(\cdot)$ can be any other function defined over $\calF$, for example, $f_p(\cdot)$ can be the ``residue distance'' of $p$, {\it i.e.,} $f_p(F)=|\dist{p}{F}-\dist{p'}{F}|$, where $p'$ is the projection of $p$ on the optimum shape $F^{\ast}$ fitting $P$. The definitions of sensitivities and total sensitivity easily carry over in this setting: $\sensitivity{\calP}{f_p}=\sup_{x \in X}f_p(x)/\sum_{f_q \in \calP}f_q(x)$ (which coincides with $\sensitivity{P}{p}$ when $f_p(\cdot)$ is $\dist{p}{\cdot}$), and $\ts{\calP}=\sum_{f_p \in \calP}\sensitivity{\calP}{f_p}$ (which coincides with $\ts{P}$ similarly). One of the results in \cite{DBLP:conf/stoc/FeldmanL11} is that an approximating subset $S \subseteq P$ can be computed with the size $|S|$ upper bounded by the product of two quantities: $(\ts{\calP})^2$, and another parameter, the ``dimension'' (see Definition \ref{defn:shatteringdim}) of a certain range space induced by $\calP$, denoted $\dim{\calP}$. We remark that $\dim{\calP}$ depends on $d$, which is the dimension of Euclidean space where $P$ is from, and some other parameters related to $X$; when $X$ is the family of shapes for the $(j,k)$-projective clustering problem, $\dim{\calP}$ also depends on $j$ and $k$. This connection allows them to use many results from the well-studied area of $\epsilon$-approximation of range spaces (such as 
deterministic construction of small $\epsilon$-approximation of range spaces), thus constructing smaller coreset deterministically, and removes some routine analysis in the traditional way of obtaining coresets via random sampling.\\

\subsection{Our Results}
In this article, we prove upper bounds of total sensitivities for the $(j,k)$-projective clustering problems. In particular, we show a careful analysis of computing total sensitivities for shape fitting problems in high dimension. Total sensitivity $\ts{P}$ for a point set $P \subset \Real^d$ may depend on $d$: consider the shape fitting problem where the family of shapes is the set of hyperplanes, and $P$ is a point set of size $d$ in general position. Then clearly $\sensitivity{P}{p}=1$ (since there always exists a hyperplane containing all $d-1$ points other than $p$), so $\ts{P}=d$. \\

One question that arises naturally is that whether the dependence of the total sensitivity on the dimension $d$ is essential. To answer this question, we show that if the distance function is Euclidean distance, or the $z^{\rm th}$ power of Euclidean distance for $z \in [1,\infty)$, then the total sensitivity function of a shape fitting problem $(\Real^d, \calF, {\rm dist})$ in the high dimensional space $\Real^d$ is roughly the same as that of the low-dimensional variant $(\Real^{d'},\calF',{\rm dist})$, where $d'$ is the ``intrinsic'' dimension of the shapes in $\calF$, and $\calF'$ consists of shapes contained in the low dimensional space $\Real^{d'}$. A reification of this statement is that the total sensitivity function of the $(j,k)$-projective clustering is independent of $d$. For the $(j,k)$-projective clustering problems, the shapes are intrinsically low dimensional: each $k$-tuple of $j$-flats is contained in a subspace of dimension at most $k(j+1)$. As we will see, the total sensitivity function 
for $(\Real^d,\calF, {\rm dist})$, where $\calF$ is the family of $k$-tuples of $j$-flats in $\Real^d$, is of the same magnitude as the total sensitivity function of $(\Real^{f(j,k)},\calF', {\rm dist})$, where $f(j,k)$ is a function of $j$ and $k$ (which is independent of $d$), and $\calF'$ is the family of $k$-tuples of $j$-flats in $\Real^{f(j,k)}$. \\

We sketch our approach to upper bound the total sensitivity of the $(j,k)$-projective clustering. We first make the observation (Theorem~\ref{thm:tsprojectedpoints} below) 
that the total sensitivity of a point set $P$ is upper bounded by a constant multiple of the total sensitivity of $P'=\proj{P}{F^{\ast}}$, which is the projection of $P$ on the optimum shape $F^{\ast}$ fitting $P$ in $\calF$. The computation of total sensitivity of $P'$ is very simple in certain cases; for example, for $k$-median clustering, $P'$ is a multi-set which contains $k$ distinct points, whose total sensitivity can be directly bounded by $k$. Therefore, we are able to greatly simplify the proofs in \cite{DBLP:conf/soda/LangbergS10}. Another more important use of this observation is that it allows us to get a dimension-reduction type result for the $(j,k)$-projective clustering problems: note that although the point set and the shapes might be in a high dimension space $\Real^d$, the projected point set $P'$ lies in a subspace of dimension $(j+1)k$ (since each $k$-tuple of $j$-flats is contained in a subspace of dimension at most $(j+1)k$), which is small under the assumption that both $j$ and $k$ 
are constant. Therefore, $\ts{P}$, which usually depends on $d$ if one directly computes it in a high dimensional space, depends only on $j$ and $k$, since $\ts{P}$ is $O(\ts{P'})$. 

Our method for bounding the total sensitivity directly translates into a template for computing $\epsilon$-coresets: 

\begin{enumerate}
\item Compute $F^{\ast}$, the optimal shape fitting $P$. (It suffices to use
an approximately optimal shape.) Compute $P'$, the projection of $P$ onto $F^{\ast}$.
\item Compute a bound on the sensitivity of each point in $P'$ with respect to $P'$. Since the ambient dimension is $O(jk)$, we may use a method that yields
bounds on $\ts{P'}$ with dependence on the ambient dimension. Use Theorem~\ref{thm:tsprojectedpoints} to translate this into a bound for
$\sensitivity{P}{p}$ for each $p \in P$.
\item Sample points from $P$  with probabilities proportional to $\sensitivity{P}{p}$ to obtain a coreset, as described in \cite{DBLP:conf/soda/LangbergS10,DBLP:conf/stoc/FeldmanL11}.
\end{enumerate}
 
We now point out the difference between our usage of total sensitivity in the construction of coresets and the method in \cite{DBLP:conf/stoc/FeldmanL11}. The construction of coresets in \cite{DBLP:conf/stoc/FeldmanL11} may also be considered as based on total sensitivity, however in a very different way: 

\begin{enumerate}
\item First obtain a small weighted point set $S \subseteq P$, such that $\dist{P}{F}-\dist{P'}{F}$ is approximately the same as $\dist{S}{F}-\dist{S'}{F}$ ($S'$ is $\proj{S}{F^{\ast}}$) for every $F \in \calF$.
\item Then compute an $\epsilon$-coreset $Q' \subseteq P'$ for the projected point set $P'$, that is, $\dist{Q'}{F}$ approximates $\dist{P'}{F}$ for every $F \in \calF$. (Since $P'$ is from a low-dimensional subspace, the ambient dimension is small, and the computation can exploit this.)
\end{enumerate} 

Therefore, for each $F \in {\cal F}$, $\dist{P}{F}  =  (\dist{P}{F}-\dist{P'}{F})+\dist{P'}{F} \approx  (\dist{S}{F}-\dist{S'}{F})+\dist{Q'}{F}$.

Thus the weighted set $Q'\cup S\cup S'$ is a coreset for $P$, but notice that the points in $S'$ have negative weights. In contrast, the weights of points in the coreset in our construction are positive. The advantage of getting coresets with positive weights is that in order to get an approximate solution to the shape fitting problem, we may run algorithms or heuristics developed for the shape fitting problem on the coreset, such as \cite{Agarwal:2004:KMP:1055558.1055581}. When points have negative weights, on the other hand, some of these heuristics do not work or need to be modified appropriately.

Another useful feature of the coresets obtained via our results is that the coreset is a subset of the original point set. When each point stands for a data item, the coreset inherits a natural interpretation. See \cite{mahoney2009matrix} for a discussion of this issue in a broader context. 
  
The sizes of the coresets in this paper are somewhat larger than the size of coresets in \cite{DBLP:conf/stoc/FeldmanL11}. Roughly speaking, the size of the coreset in \cite{DBLP:conf/stoc/FeldmanL11} is $f_1(d)+f_2(j,k)$, where $f_1(d)$ (respectively $f_2(j,k)$) is a function depending only on $d$ (respectively $j$ and $k$) for the $(j,k)$-projective clustering problem, while the coreset size in our paper is $f_1(d)\cdot f_2(j,k)$. 

{\bf \noindent Organization of this paper:} In this article, we focus on the construction that establishes small total sensitivity for various shape fitting problems, and the size of the resulting coreset. For clarity, we omit the description of algorithms for computing such bounds on sensitivity. Efficient algorithms result from the construction using a methodology that is now well-understood. Also because the weights for points in the coreset are nonnegative, the coreset lend itself to streaming settings, where points arrive one by one as $p_1,p_2,\cdots$~\cite{Har-Peled:2004:CKK:1007352.1007400}\cite{DBLP:conf/stoc/FeldmanL11}. In Section 2, we present necessary definitions used through this article, and summarize related results from \cite{DBLP:conf/stoc/FeldmanL11} and \cite{DBLP:conf/soda/VaradarajanX12}. In Section 3, we prove the upper bound of total sensitivity of an instance of a shape fitting problem in high dimension by its low dimensional projection. In Sections 4, 5, 6, and 7, we apply the 
upper bound from Section 3 to $k$-median/$k$-means, clustering, $k$-line clustering, $j$-subspace approximation, and the integer $(j,k)$-projective clustering problem, respectively, to obtain upper bounds for their total sensitivities, and the 
size of the resulting $\epsilon$-coresets.

\section{Preliminaries}
\label{sec:prelims}
In this section, we formally define some of the concepts studied in this article, and state crucial results from previous work. We begin by defining an $\epsilon$-coreset. 


\begin{definition}[$\epsilon$-coreset of a shape fitting problem]
\label{defn: epsiloncoreset}
Given an instance $P \subset \Real^d$ of a shape fitting problem $(\Real^d, \calF, {\rm dist})$, and $\epsilon \in [0,1]$, an $\epsilon$-coreset of $P$ is a (weighted) set $S \subseteq P$, together with a weight function $w: S \to \Real^+$, such that for any shape $F$ in $\calF$, it holds that $\abs{\dist{P}{F}-\dist{S}{F}} \leq \epsilon \cdot \dist{P}{F}$,
where by definition, $\dist{P}{F}  = \sum_{p \in P}\dist{p}{F}, \mbox{ and } \dist{S}{F}  = \sum_{p \in S} w(p)\dist{p}{F}$.
The size of the weighted coreset $S$ is defined to be $|S|$.
\end{definition}

We note that in the literature, the requirement that the weights be non-negative, as well as the requirement that the coreset $S$ be a subset of the original instance $P$, are sometimes relaxed. We include these requirements in the definition to emphasize that the coresets constructed here do satisfy them. We now define the sensitivities of points in a shape fitting instance, and the total sensitivity of the instance.  

\begin{definition}[Sensitivity of a shape fitting instance~\cite{DBLP:conf/soda/LangbergS10}]
\label{defn:sensitivity}
Given an instance $P \subset \Real^d$ of a shape fitting problem $(\Real^d, \calF, {\rm dist})$, the sensitivity of a point $p$ in $P$ is
$\sensitivity{P}{p}:=\inf\{\beta \geq 0 \vert \dist{p}{F}\leq \beta \dist{P}{F},\forall F \in \calF\}.$

Note that an equivalent definition is to let $\sensitivity{P}{p} =
\sup_{F \in \calF} \dist{p}{F}/\dist{P}{F}$, with the understanding that when the denominator in the ratio is $0$, the ratio itself is $0$.

The total sensitivity of the instance $P$, is defined by $\ts{P}:=\sum_{p \in P}\sensitivity{P}{p}$.
The total sensitivity function of the shape fitting problem is $\ts{n}:=\sup_{\abs{P}=n} \ts{P}$.
\end{definition}

We now need a somewhat technical definition in order to be able to state an important earlier result from \cite{DBLP:conf/stoc/FeldmanL11}. On a first reading, the reader is welcome to skip the detailed definition.

\begin{definition}[The dimension of a shape fitting instance~\cite{DBLP:conf/stoc/FeldmanL11}]
\label{defn:shatteringdim}
Let $P \subset \Real^d$ be an instance of a shape fitting problem $(\Real^d, \calF, {\rm dist})$. For a weight function $w: P \to \Real^+$, consider the set system $(P, \calR)$, where $\calR$ is a family of subsets of $P$ defined as follows: each element in $\calR$ is a set of the form $R_{F,r}$ for some $F \in \calF$ and $r \geq 0$, and $R_{F,r}=\{p \in P \ \vert \ w_p \cdot \dist{p}{F} \leq r\}$. That is, $R_{F,r}$ is the set of those points in $P$ whose weighted distance to the shape $F$ is at most $r$. 
The dimension of the instance $P$ of the shape fitting problem, denoted by $\dim{P}$, is the smallest integer $m$, such that for any weight function $w$ and $A \subseteq P$ of size $\abs{A}=a \geq 2$, we have: $\abs{\{A \cap R_{F,r} \vert F \in \calF, r \geq 0\}} \leq a^m$.
\end{definition}

For instance, in the $(j,k)$-projective clustering problem with the underlying distance function ${\rm dist}$ being the $z^{\rm th}$ power of the Euclidean distance, the dimension $\dim{P}$ of any instance $P$ is $O(jdk)$, independent of $|P|$ \cite{DBLP:conf/stoc/FeldmanL11}. This is shown by methods similar to the ones
 used to bound the VC-dimension of geometric set systems. In fact, this bound is the only fact that we will need about the dimension of a shape fitting instance. 

The following theorem recalls the connection established in \cite{DBLP:conf/stoc/FeldmanL11} between coresets and sensitivity via the above notion of dimension.

\begin{theorem}[Connection between total sensitivity and $\epsilon$-coreset~\cite{DBLP:conf/stoc/FeldmanL11}]
\label{ts:connectionbetweentotalsensitivityandcoreset}
Given any $n$-point instance $P \subset \Real^d$ of a shape fitting problem 
$(\Real^d, \calF, {\rm dist})$, and any $\epsilon \in (0, 1]$, there exists an $\epsilon$-coreset for $P$ of size $O\left(\left(\frac{\ts{n}}{\epsilon}\right)^2\dim{P}\right)$.
 \end{theorem}

Finally, we will need known bounds on the total sensitivity of $(j,k)$-projective clustering problem. These earlier bounds involve the dimension $d$ corresponding to shape fitting problem $(\Real^d,\calF, {\rm dist})$.
 
\begin{theorem}[Total sensitivity of $(j,k)$-projective clustering problem in fixed dimension~\cite{DBLP:conf/soda/VaradarajanX12}
]
\label{thm:tsklinecenterintegerjkprojectiveclustering}
We have the following upper bounds of total sensitivities for the $(j,k)$-projective clustering problem $(\Real^d, \calF, {\rm dist})$, where ${\rm dist}$ is the $z$-th power of the Eucldiean distance for $z \in (0, \infty)$.
\begin{itemize}
 \item $j=1$ ($k$-line center): $\ts{n}$ is $O(k^{f(k,d)} \log n)$, where $f(d,k)$ is a function depending only on $d$ and $k$.
 \item integer $(j,k)$-projective clustering problem: For any $n$-point instance $P$, with each coordinate being an integer of magnitude at most $n^{c}$ for
any constant $c > 0$, $\ts{P}$ is $O((\log n)^{f(d,j,k)})$, where $f(d,j,k)$ is a function depending only on $d$, $j$, and $k$.
\end{itemize}
\end{theorem}

\section{Bounding the Total Sensitivity via Dimension Reduction}
In this section, we show that the total sensitivity of a point set $P$ is of the same order as that of $\proj{P}{F^{\ast}}$, which is the projection of $P$ onto an optimum shape $F^{\ast}$ from $\calF$ fitting $P$. This result captures the fact that total sensitivity of a shape fitting problem quantifies the complexity of shapes, in the sense that total sensitivity depends on the dimension of smallest subspace containing each shape, regardless of the dimension of the ambient space where $P$ is from. 
\begin{definition}[projection of points on a shape]
For a shape fitting problem $(\Real^d, \calF, {\rm dist})$,
define ${\rm proj}: \Real^d \times \calF \to \Real^d$, where $\proj{p}{F}$ is the projection of $p$ on a shape $F$, that is, $\proj{p}{F}$ is a point in $F$ which is nearest to $p$, with ties broken arbitrarily. That is, $\dist{p}{\proj{p}{F}}=\min_{q \in F}\dist{p}{q}$. We abuse the notation to denote the multi-set $\{\proj{p}{F} \vert p \in P\}$ by $\proj{P}{F}$ for $P \subset \Real^d$.
\end{definition}
We first show that $\ts{P}$ is $O(\ts{\proj{P}{F^{\ast}}})$, where $F^{\ast}$ is an optimum shape fitting $P$ from $\calF$. In particular, this implies that when $F^{\ast}$ is a low-dimensional object, the total sensitivity of $P \subset \Real^d$ can be upper bounded by the total sensitivity of a point set contained in a low dimension subspace.
\begin{theorem}[Dimension reduction, computing the total sensitivity of a point set in high dimensional space with the projected lower dimensional point set]
\label{thm:tsprojectedpoints}
Given an instance $P$ of a shape fitting problem $(\Real^d, \calF, {\rm dist})$, let $F^{\ast}$ denote a shape that minimizes $\dist{P}{F}$ over all $F \in \calF$. Let $p'$ denote $\proj{p}{F^{\ast}}$ and let $P'$ denote $\proj{P}{F^{\ast}}$. Assume that the distance function satisfies the relaxed triangle inequality: $\dist{p}{q} \leq \alpha (\dist{p}{r}+\dist{r}{q})$ for any $p,q,r \in \Real^d$ for some constant $\alpha \geq 1$. Then
\begin{enumerate}
\item the following inequality holds: $\ts{P} \leq 2\alpha^2\ts{P'}+\alpha$.
\item if $\dist{P}{F^{\ast}} = 0$, then $\sensitivity{P}{p} = \sensitivity{P'}{p'}$ for each $p \in P$. If $\dist{P}{F^{\ast}} > 0$, then $\sensitivity{P}{p} \leq \left( \alpha \frac{ \dist{p}{p'} }{ \dist{P}{F^{\ast}}} + 2 \alpha^2 \sensitivity{P'}{p'} \right).$
\end{enumerate}  
\end{theorem}

\begin{proof}
If $\dist{P}{F^{\ast}} = 0$, then $P = P'$, and clearly both parts of the theorem hold.

Let us consider the case where $\dist{P}{F^{\ast}} > 0$. By definition, 
\begin{align*}
 \sensitivity{P}{p}& =\inf\{\beta \geq 0 \ \vert \ \dist{p}{F} \leq \beta \dist{P}{F}, \forall F \in \calF\}, \\
 \sensitivity{P'}{p'}&=\inf\{\beta' \geq 0 \ \vert \ \dist{p'}{F}\leq \beta' \dist{P'}{F},\forall F \in \calF\}.
\end{align*}
Let $F$ be an arbitrary shape in $\calF$. Then we have
\begin{equation*}
 \begin{split}
 \dist{p}{F} & \leq \alpha \dist{p}{p'}+ \alpha \dist{p'}{F}  \\
             & \leq \alpha \dist{p}{p'}+ \alpha \sensitivity{P'}{p'}\dist{P'}{F} \\
             & \leq \alpha \dist{p}{p'}+ 2\alpha^2 \sensitivity{P'}{p'}\dist{P}{F}\\ 
             & = \alpha \frac{\dist{p}{p'}}{\dist{P}{F}} \cdot \dist{P}{F} + 2 \alpha^2 \sensitivity{P'}{p'}\dist{P}{F} \\
             & \leq \alpha \frac{\dist{p}{p'}}{\dist{P}{F^{\ast}}} \cdot \dist{P}{F} + 2 \alpha^2 \sensitivity{P'}{p'}\dist{P}{F}\\
             & = \left( \alpha \frac{\dist{p}{p'}}{\dist{P}{F^{\ast}}} +  2 \alpha^2 \sensitivity{P'}{p'} \right) \dist{P}{F}.
 \end{split}
\end{equation*}
The first inequality follows from the relaxed triangle inequality, the second inequality follows from the definition of sensitivity of $p'$ in $P'$, and third inequality follows from the fact that $\dist{P'}{F} = \sum_{p' \in P'} \dist{p'}{F} \leq \sum_{p \in P}\alpha\left(\dist{p}{F}+\dist{p}{p'}\right) = \alpha(\dist{P}{F}+\dist{P}{F^{\ast}}) \leq 2\alpha\dist{P}{F},$
since $\dist{P}{F^{\ast}} \leq \dist{P}{F}$. 

Thus the second part of the theorem holds. Now,
\begin{eqnarray*}
\ts{P} & = & \sum_{p \in P} \sensitivity{P}{p} \\
       & \leq & \sum_{p \in P} \left( \alpha \frac{\dist{p}{p'}}{\dist{P}{F^{\ast}}} +  2 \alpha^2 \sensitivity{P'}{p'} \right) \\
       & = & \alpha + 2 \alpha^2 \ts{P'}.
\end{eqnarray*}
\end{proof}

We make a remark regarding the value of $\alpha$ in Theorem~\ref{thm:tsprojectedpoints} when the distance function is $z^{\rm th}$ power of Euclidean distance. It is used in Sections 4, 5, 6, and 7 when we derive upper bounds of total sensitivities for various shape fitting problems.
\begin{remark}[Value of $\alpha$ when $\dist{\cdot}{\cdot}=(\pnorm{\cdot}{2})^z$]
Let $z \in (0,\infty)$. Suppose $\dist{p}{q}=(\pnorm{p-q}{2})^z$. When $z\in (0,1)$, the weak triangle inequality holds with $\alpha=1$; when $z \geq 1$, the weak triangle inequality holds with $\alpha=2^{z-1}$. For a proof, see, for example, \cite{DBLP:conf/soda/FeldmanS12}.  
\end{remark}
Theorem~\ref{thm:tsprojectedpoints} bounds the total sensitivity of an instance
$P$ of a shape fitting problem $(\Real^d, \calF, {\rm dist})$ in terms of the total sensitivity of $P'$. Suppose that there is an $m_2 \ll d$ so that each shape $F \in \calF$ is in some subspace of dimension $m_2$. In the $(j,k)$-projective clustering problem, for example, $m_2 = k(j+1)$. Then note that $P'$ is contained in a subspace of dimension $m_2$. Furthermore, when ${\rm dist}$ is the $z^{\rm th}$ power of the Euclidean distance, it turns out that for many shape fitting problems the sensitivity of $P'$ can be bounded as if the shape fitting problem was housed in $\Real^{2m_2}$ instead of $\Real^d$. To see why this is the case for the $(j, k)$-projective clustering problem, fix an arbitrary subspace $G$ of 
dimension $\min \{d, 2m_2 \}$ that contains $P'$. Then for for any $F \in \calF$, there is an $F' \in \calF$ such that (a) $F'$ is contained in $G$, and (b)
$\dist{p'}{F'} = \dist{p'}{F}$ for all $p' \in P'$. 

The following theorem summarizes this phenomenon. For simplicity, it is stated  for the $(j,k)$-projective clustering problem, even though the phenomenon itself is somewhat more general. 

\begin{theorem}[Sensitivity of a lower dimensional point set in a high dimensional space]
\label{thm:dimensionreduction}
Let $P'$ be an $n$-point instance of the $(j,k)$-projective clustering problem 
$(\Real^d, \calF, {\rm dist})$, where ${\rm dist}$ is the $z^{\rm th}$ power of the Euclidean distance, for some $z \in (0,\infty)$. Assume that $P'$ is contained in a subspace of dimension $m_1$. (Note that for each shape $F \in \calF$,  there is a subspace of dimension $m_2 = k(j+1)$ containing it.) Let $G$ be  any subspace of dimension $m = \min \{m_1 + m_2, d\} $ containing $P'$; fix an orthonormal basis for $G$, and for each $p' \in P'$, let $p'' \in \Real^m$ be the coordinates of $p'$ in terms of this basis. Let $P'' = \{p'' \ \mid \ p' \in P'\}$, and view $P''$ as an instance of the $(j,k)$-projective clustering problem 
$(\Real^m, \calF', {\rm dist})$, where $\calF'$ is the set of all $k$-tuples of
$j$-subspaces in $\Real^m$, and ${\rm dist}$ is the $z^{\rm th}$ power of the Eucldiean distance. Then, $\sensitivity{P'}{p'} = \sensitivity{P''}{p''}$ for each
$p' \in P'$, and $\ts{P'} = \ts{P''}$.  
\end{theorem}

\section{$k$-median/$k$-means Clustering Problem}
In this section, we derive upper bounds for the total sensitivity function for the $k$-median/$k$-means problems, and its generalizations, where the distance function is $z^{\rm th}$ power of Euclidean distance, using the approach in Section 4. These bounds are similar to the ones derived by Langberg and Schulman~\cite{DBLP:conf/soda/LangbergS10}, but the proof is much simplified. For the rest of the article, ${\rm dist}$ is assumed to be the $z^{\rm th}$ power of the Euclidean distance.
\begin{theorem}[Total sensitivity of $(0,k)$-projective clustering]
\label{thm:kmediants}
Consider the shape fitting problem $(\Real^d, \calF, {\rm dist})$, where $\calF$ is the set of all $k$-point subsets of $\Real^d$. We have the following upper bound on the total sensitivity: 
\begin{align*}
 \ts{n} & \leq 2^{2z-1} k + 2^{z-1}, & z \geq 1,\\
 \ts{n} & \leq 2k+1, & z \in (0,1).
\end{align*}
In particular, the total sensitivity of the $k$-median problem (which corresponds to the case when $z=1$) is at most $2k+1$, and the total sensitivity of the $k$-means problem (which corresponds to the case when $z=2$) is $8k+2$.
\end{theorem}
\begin{proof}
Let $P$ be an arbitrary $n$-point set. Apply Theorem~\ref{thm:tsprojectedpoints}, and note that $\proj{P}{C^{\ast}}$, where $C^{\ast}$ is an optimum set of $k$ centers, contains at most $k$ distinct points. Assume that $C^{\ast}=\{c_1^{\ast}, c_2^{\ast},\cdots, c_k^{\ast}\}$. Let $P_i$ be the set of points in $P$ whose projection is $c_i^{\ast}$, that is, $P_i = \{p \in P \vert \proj{p}{C^{\ast}}=c_i^{\ast}\}$. It is easy to see that the summation of sensitivities of the $\abs{P_i}$ copies of $c_i^{\ast}$ is at most 1: for any $k$-point set $C$ in $\Real^d$, 
$\abs{P_i}\cdot \frac{\dist{c_i^{\ast}}{C}}{\dist{C^{\ast}}{C}}=\frac{\abs{P_i}\dist{c_i^{\ast}}{C}}{\sum_{j=1}^k \abs{P_j}\dist{c_j^{\ast}}{C}} \leq 1$.

Therefore, the total sensitivity of $\proj{P}{C^{\ast}}$ is at most $k$. Substituting $\alpha$ from the remark after Theorem~\ref{thm:tsprojectedpoints}, we get the above result.
\end{proof}
\begin{theorem}[$\epsilon$-coreset for $(0,k)$-projective clustering]
Consider the shape fitting problem $(\Real^d, \calF, {\rm dist})$, where $\calF$ is the set of all $k$-point subsets of $\Real^d$. For any $n$-point instance $P$, there is an $\epsilon$-coreset of size $O(k^3 d \epsilon^{-2})$.
\end{theorem}
\begin{proof}
Observe that the $\dim{P}$ is $O(kd)$. Using Theorem~\ref{ts:connectionbetweentotalsensitivityandcoreset}, and Theorem~\ref{thm:kmediants}, we obtain the above result. 
\end{proof}

\section{$k$-line Clustering Problem}
In this section, we derive upper bounds on the total sensitivity function for the $k$-line clustering problem, that is, the $(1,k)$-projective clustering problem.  
\begin{theorem}[Total sensitivity for $k$-line clustering problem]
\label{thm:tsklinecenter}
Consider the shape fitting problem $(\Real^d, \calF, {\rm dist})$, where
$\calF$ is the set of $k$-tuple of lines. The total sensitivity function, $\ts{n}$, is $O(k^{f(k)} \log n)$, where $f(k)$ is a function the depends only on $k$.
\end{theorem}
\begin{proof}
Let $P$ be an arbitrary $n$-point set. Let $K^{\ast}$ denote an optimum set of $k$ lines fitting $P$. Using Theorems~\ref{thm:tsprojectedpoints} and \ref{thm:dimensionreduction}, it suffices to bound the sensitivity of an $n$-point instance
of a $k$-line clustering problem housed in $\Real^{4k}$. By Theorem~\ref{thm:tsklinecenterintegerjkprojectiveclustering}, the total sensitivity of this latter shape fitting problem is $O(k^{f(k)} \log n)$, where $f(k)$ is a function depending only on $k$. Therefore, $\ts{n}$ is $O(k^{f(k)} \log n)$.

(Alternatively, one could use a recent result in \cite{DBLP:conf/soda/FeldmanS12}. Let $P'$ denote the projection of $P$ into $K^{\ast}$. Since $K^{\ast}$ is a union of $k$ lines, we can upper bound the sensitivity of $P'$ by $k$ times the 
sensitivity of an $n$-point set that lies on a {\em single line}. The sensitivity of an $n$-point set that lies on a single line can be upper bounded by the
sensitivity of an $n$-point set for the {\em weighted} $(0,k)$-projective clustering problem, for which the sensitivity bound is $O(k^{f(k)}\log n)$ as shown in \cite{DBLP:conf/soda/FeldmanS12}.)
\end{proof}
Notice that for $k$-line clustering problem, the bound on the total sensitivity depends logarithmically on $n$. We give below a construction of a point set that shows that this is necessary, even for $d=2$. 
\begin{theorem}[The upper bound of total sensitivity for $k$-line clustering problem is tight]
For every $n \geq 2$, there exists an $n$-point instance of the $k$-line clustering problem $(\Real^2, \calF, {\rm dist})$, where ${\rm dist}$ is the Euclidean distance, such that the total sensitivity of $P$ is $\Omega(\log n)$.
\end{theorem}
\begin{proof}
 We construct a point set $P$ of size $n$, together with $n$ shapes $F_i \in \calF$, $i=1,\cdots,n$, such that $\sum_{i=1}^{n}\dist{p_i}{F_i}/\dist{P}{F_i}$ is $\Omega(\log n)$. Note that this implies that $\ts{P}$ is at least $\Omega(\log n)$. Let $P$ be the following point set in $\Real^2$: $p_i=(1/2^{i-1},0)$, for $i=1,\cdots, n$. Let $F_i$ be a pair of lines: one vertical line and one horizontal line, where the vertical line is the $y$-axis, and the horizontal line is 
$\{(x,1/2^{i}) \vert x \in \Real\}$. 

 Consider the point $p_i$, where $i=1,\cdots, n$. We show that $\dist{p_i}{F_i}/\dist{P}{F_i}$ is at least $1/(2+i)$, for $i=1,\cdots, n$. For $j\leq i$, note that $\dist{p_j}{F_i}=1/2^i$: since the distance from $p_j$ to the horizontal line in $F_i$ is $1/2^i$ and the distance to the vertical line is $1/2^{j-1}$, $\dist{p_j}{F_i} = \min\{1/2^{j-1}, 1/2^i\}=1/2^i$. For $i+1\leq j\leq n$, on the other hand, $\dist{p_j}{F_i}=1/2^{j-1}$. Therefore, $\sum_{j=i+1}^n \dist{p_j}{F_i}=\sum_{j=i+1}^n 1/2^{j-1} =(1/2^{i-1})\cdot (1-(1/2)^{n-i})$. Thus, we have
 \[
  \sensitivity{P}{p_i}=\sup_{F \in \calF}\frac{\dist{p_i}{F}}{\dist{P}{F}} \geq \frac{\dist{p_i}{F_i}}{\dist{P}{F_i}} = \frac{1/2^i}{(1/2^{i-1}-1/2^{n-1}) + i \cdot (1/2^i)} > \frac{1}{2+i}
 \]
Therefore, $\ts{P} \geq \sum_{i=1}^{n}\sensitivity{P}{p_i} >\sum_{i=1}^{n}\frac{1}{2+i}$, which is $\Omega(\log n)$.  
\end{proof}
\begin{theorem}[$\epsilon$-coreset for $k$-line clustering problem]
Consider the shape fitting problem $(\Real^d, \calF, {\rm dist})$, where $\calF$ is the set of all $k$-tuples of lines in $\Real^d$. For any $n$-point instance $P$, there is an $\epsilon$-coreset with size $O(k^{f(k)}d(\log n)^2/\epsilon^2)$.
\end{theorem}
\begin{proof}
 This result follows from Theorem~\ref{thm:tsklinecenter}, Theorem~\ref{ts:connectionbetweentotalsensitivityandcoreset}, and the fact that $\dim{P}$ in this case is $O(kd)$. 
\end{proof}

\section{Subspace approximation}
In this section, we derive upper bounds on the sensitivity of the subspace approximation problem, that is, the $(j,1)$-projective clustering problem. For the applications of Theorems~\ref{thm:tsprojectedpoints} and \ref{thm:dimensionreduction} in the other sections, we use existing bounds on the sensitivity that have a dependence on the dimension $d$. For the subspace approximation problem, however, we derive here the dimension-dependent bounds on sensitivity by generalizing an argument from \cite{DBLP:conf/soda/LangbergS10} for the case $j = d-1$ and
$z = 2$. This derivation is somewhat technical. With these bounds in hand, the derivation of the dimension-independent bounds is readily accomplished in a manner similar to the other sections.

distance. Although the size of the $\epsilon$-coreset obtained in this way is exponential in $j$, which is larger than the size of the coreset in \cite{DBLP:conf/stoc/FeldmanL11}\cite{feldmanarvix} and Theorem~\ref{} in this section, it is still a constant (as $j$ is considered as a constant) and in particular, independent of the cardinality of the input point set. It can be considered as an simple and straight-forward way to see why small $\epsilon$-coresets exist for $j$-subspace approximation problems.\\

\subsection{Dimension-dependent bounds on Sensitivity}
We first recall the notion of an \emph{$(\alpha,\beta,z)$-conditioned basis} from \cite{DBLP:journals/siamcomp/DasguptaDHKM09}, and state one of its properties (Lemma~\ref{lem:linearalgebraresult}). We will use standard matrix terminilogy: $m_{ij}$ denotes the entry in the $i$-th row and $j$-th column of $M$, and $M_{i\cdot}$ is the $i$-th row of $M$.
 
\begin{definition}
 Let $M$ be an $n \times m$ matrix of rank $\rho$. Let $z \in [1,\infty)$, and $\alpha, \beta \geq 1$. An $n \times \rho$ matrix $A$ is an $(\alpha, \beta, z)$-conditioned basis for $M$ if the column vectors of $A$ span the column space of $M$, and additionally $A$ satisfies that: (1) $\sum_{i,j} \abs{a_{ij}}^z \leq \alpha^z$, (2) for all $u \in \Real^{\rho}$, $\pnorm{u}{z'} \leq \beta \pnorm{Au}{z}$, where $\pnorm{\cdot}{z'}$ is the dual norm for $\pnorm{\cdot}{z}$ ({\it i.e.} $1/z+1/z'=1$).
\end{definition}

\begin{lemma}
\label{lem:linearalgebraresult}
Let $M$ be an $n \times m$ matrix of rank $\rho$. Let $z \in [1,\infty)$. Let $A$ be an $(\alpha, \beta, z)$-conditioned basis for $M$. For every vector $u \in \Real^m$, the following inequality holds: $\abs{M_{i\cdot}u}^z \leq \left(\pnorm{A_{i\cdot}}{z}^z \cdot \beta^z\right)\pnorm{M u}{z}^z$.

\end{lemma}
\begin{proof}
We have $M=A\tau$ for some $\rho \times m$ matrix $\tau$. Then,
\begin{equation*}
\abs{M_{i\cdot}u}^z = \abs{A_{i\cdot}\tau u}^z \leq \pnorm{A_{i\cdot}}{z}^z \cdot \pnorm{\tau u}{z'}^z \leq \pnorm{A_{i\cdot}}{z}^z \cdot \beta^z \pnorm{A\tau u}{z}^z  = \pnorm{A_{i\cdot}}{z}^z \cdot \beta^z \pnorm{M u}{z}^z.
\end{equation*}
The second step is Holder's inequality, and the third uses the fact that $A$ is 
$(\alpha, \beta, z)$-conditioned.
\end{proof}

Using Lemma~\ref{lem:linearalgebraresult}, we derive an upper bound on the total sensitivity when each shape is a hyperplane.

\begin{lemma}[total sensitivity for fitting a hyperplane]
\label{lem: tshyperplane}
Consider the shape fitting problem $(\Real^d, \calF, {\rm dist})$ where $\calF$ is the set of all $(d-1)$-flats, that is, hyperplanes.
The total sensitivity of any $n$-point set is $O(d^{1+z/2})$ for $1 \leq z < 2$, 
$O(d)$ for $z=2$, and $O(d^z)$ for $z>2$.
\end{lemma}
\begin{proof}
We can parameterize a hyperplane with a vector in $\Real^{d+1}$, $u=\begin{bmatrix} u_1 & \cdots & u_{d+1} \end{bmatrix}^{T}$: the hyperplane determined by $u$ is $h_{u}=\{x \in \Real^d \vert \sum_{i=1}^d u_i x_i + u_{d+1}=0\}$, where $x_i$ denotes the $i^{\rm th}$ entry of the vector $x$. Without loss of generality, we may assume that $\sum_{i=1}^d u_i^2=1$. The Euclidean distance to $h_u$ from a point $q \in \Real^d$ is $\dist{q}{h_{u}}=\abs{\sum_{i=1}^d u_i q_i+u_{d+1}}/\sqrt{\sum_{i=1}^d u_i^2}=\abs{\sum_{i=1}^d u_i q_i+u_{d+1}}.$
(the second equality follows from the assumption that $\sum_{i=1}^d u_i^2=1$.) 

Let $P = \{p_1, p_2, \ldots, p_n\} \subseteq \Real^d$ be any set of $n$ points. 
Let $\tilde{p_i}$ denote the row vector $\begin{bmatrix}  p_i^T & 1\end{bmatrix}$, and let $M$ be the $n \times (d+1)$ matrix whose $i^{\rm th}$ row is $\tilde{p_i}$. Then, $\dist{p_i}{h_u} = \abs{M_{i\cdot}u}^z$, and $\dist{P}{h_u} = \sum_{i=1}^n \abs{M_{i\cdot}u}^z = \pnorm{M u}{z}^z$. Then using Lemma~\ref{lem:linearalgebraresult}, we have $\sensitivity{P}{p_i} = \sup_u \frac{\abs{M_{i\cdot}u}^z}{\pnorm{M u}{z}^z} \leq
\pnorm{A_{i\cdot}}{z}^z \cdot \beta^z$,
where $A$ is an $(\alpha, \beta, z)$-conditioned basis for $M$. Thus,
\[ \ts{P} = \sum_{i=1}^n \sensitivity{P}{p_i}  \leq \beta^z \sum_{i=1}^n \pnorm{A_{i\cdot}}{z}^z = \beta^z \sum_{i,j}\abs{a_{ij}}^z= (\alpha \beta)^z.\] 
For $1 \leq z<2$, $M$ has $((d+1)^{1/z+1/2},1,z)$-conditioned basis; for $z=2$, $M$ has $((d+1)^{1/2},1,z)$-conditioned basis; for $z>2$, $M$ has $((d+1)^{1/z+1/2}, (d+1)^{1/z'-1/2},z)$-conditioned basis \cite{DBLP:journals/siamcomp/DasguptaDHKM09}. Thus the total sensitivity for the three cases are $(d+1)^{1+z/2}$, $d+1$, and $(d+1)^{z}$, respectively.
\end{proof}

It is now easy to derive dimension-dependent bounds on the sensitivity when each shape is a $j$-subspace.

\begin{corollary}[Total sensitivity for fitting a $j$-subspace]
\label{coro:jsubspace}
Consider the shape fitting problem $(\Real^d, \calF, {\rm dist})$ where $\calF$ is the set of all $j$-flats.
The total sensitivity of any $n$-point set is $O(d^{1+z/2})$ for $1 \leq z < 2$, 
$O(d)$ for $z=2$, and $O(d^z)$ for $z>2$.
\end{corollary}
\begin{proof}
Denote by $\calF'$ the set of hyperplanes in $\Real^d$. Let $P \subseteq \Real^d$ be an arbitrary $n$-point set. We first show that $\sensitivity{P,\calF}{p} \leq \sensitivity{P,\calF'}{p}$, where the additional subscript is being used to indicate which shape fitting problem we are talking about (hyperplanes or $j$-flats). Let $p$ be an arbitrary point in $P$. Let $F_p \in \calF$ denote the $j$-subspace such that $\sensitivity{P,\calF}{p}=\dist{p}{F_p}/\dist{P}{F_p}$. Let $\proj{p}{F_p}$ denote the projection of $p$ on $F_p$. Consider the hyperplane $F'$ containing $F_p$ and orthogonal to the vector $p-\proj{p}{F_p}$. We have $\dist{p}{F'}=\dist{p}{F_p}$, whereas $\dist{q}{F'} \leq \dist{q}{F_p}$ for each $q \in P$. Therefore, $\sensitivity{P,\calF'}{p} \geq \dist{p}{F'}/\dist{P}{F'} \geq \dist{p}{F_p}/\dist{P}{F_p} = \sensitivity{P,\calF}{p}.$
It follows that $\ts{P,\calF} \leq \ts{P,\calF'}$. The statement in the corollary now follows from Lemma~\ref{lem: tshyperplane}.
\end{proof}

\subsection{Dimension-independent Bounds on the Sensitivity}

We now derive dimension-independent upper bounds for the  total sensitivity for the $j$-subspace fitting problem.
\begin{theorem}[Total sensitivity for $j$-subspace fitting problem]
\label{thm:jsubspacetotalsensitivity}
Consider the shape fitting problem $(\Real^d, \calF, {\rm dist})$ where $\calF$ is the set of all $j$-flats.
The total sensitivity of any $n$-point set is $O(j^{1+z/2})$ for $1 \leq z < 2$, 
$O(j)$ for $z=2$, and $O(j^z)$ for $z>2$.
\end{theorem}
\begin{proof}
 Use Theorem~\ref{thm:tsprojectedpoints}, note that the projected point set $P'$ is contained in a $j$-subspace. Further, each shape is a $j$-subspace. So, applying Theorem~\ref{thm:dimensionreduction} and Corollary~\ref{coro:jsubspace}, the total sensitivity is $O(j^{2+z/2})$ or $z\in [1,2)$, $O(j)$ for $z=2$ and $O(j^{z})$ for $z > 2$.
\end{proof}

Using Theorem~\ref{thm:jsubspacetotalsensitivity} and the fact that $\dim{P}$ for the $j$-subspace fitting problem is $O(jd)$, we obtain small $\epsilon$-coresets:
\begin{theorem}[$\epsilon$-coreset for $j$-subspace fitting problem]
Consider the shape fitting problem $(\Real^d, \calF, {\rm dist})$ where $\calF$ is the set of all $j$-flats. For any $n$-point set, there exists an $\epsilon$-coreset whose size is $O(j^{3+z}d\epsilon^{-2})$ for $z\in [1,2)$, $O(j^3d\epsilon^{-2})$ for $z=2$ and $O(j^{2z+1}d\epsilon^{-2})$ for $z \geq 2$.  
\end{theorem}
\begin{proof}
 The result follows from Theorem~\ref{thm:jsubspacetotalsensitivity}, and Theorem~\ref{ts:connectionbetweentotalsensitivityandcoreset}.
\end{proof}

We note that for the case $j = d-1$ and $z = 2$, a linear algebraic result from~\cite{DBLP:conf/stoc/BatsonSS09} yields a coreset whose size is an improved 
$O(d \epsilon^{-2})$.  

\section{The $(j,k)$ integer projective clustering}
\begin{theorem}
Consider the shape fitting problem $(\Real^d, \calF, {\rm dist})$, where
$\calF$ is the set of $k$-tuples of $j$-flats. Let $P \subset \Real^d$ be any $n$-point instance with integer coordinates, the magnitude of each coordinate being at most $n^c$, for some constant $c$. The total sensitivity $\ts{P}$ of 
$P$ is $O((\log n)^{f(k,j)})$, where $f(k,j)$ is a function of only $k$ and $j$.  There exists an $\epsilon$-coreset for $P$ of size $O((\log n)^{2f(k,j)}kjd\epsilon^{-2})$. 
\end{theorem}
\begin{proof}
 Observe that the projected point set $P'=\proj{P}{\{J_1^{\ast},\cdots, J_k^{\ast}\}}$, where $\{J_1^{\ast},\cdots, J_k^{\ast}\}$ is an optimum $k$-tuple of $j$-flats fitting $P$, is contained in a subspace of dimension $O(jk)$. Using Theorem~\ref{thm:tsklinecenterintegerjkprojectiveclustering}, Theorem~\ref{thm:dimensionreduction}, and Theorem~\ref{thm:tsprojectedpoints}, the total sensitivity $\ts{P}$ is upper bouned by $O((\log n)^{f(k,j)})$, where $f(k,j)$ is a function of $k$ and $j$.  (A technical complication is that the coordinates of $P'$, in the appropriate orthonormal basis, may not be integers. This can be addressed by rounding them to integers, at the expense of increasing the constant $c$. A similar procedure is adopted in \cite{DBLP:conf/soda/VaradarajanX12}, and we omit the details here.)

Using Theorem~\ref{ts:connectionbetweentotalsensitivityandcoreset} and the fact that $\dim{P}$ is $O(djk)$, we obtain the bound on the coreset.
\end{proof}

\section{Acknowledgements.} We thank the anonymous reviewers and Dan Feldman for their insightful feedback.

\bibliography{On_the_Sensitivity_of_Shape_Fitting_Problems}

\end{document}